\documentclass{article}

\usepackage[utf8]{inputenc}
\usepackage{hyperref}
\usepackage[super,sort,compress]{cite}
\usepackage{authblk} 
\usepackage{amsmath} 
\usepackage{amsfonts} 
\usepackage{amsthm} 
\usepackage{braket}

\theoremstyle{plain}
\newtheorem{theorem}{Theorem}
\newtheorem{proposition}{Proposition}
\newtheorem{corollary}{Corollary}
\newtheorem{lemma}{Lemma}

\theoremstyle{definition}
\newtheorem{example}{Example}

\newcommand{\hil}{\mathcal{H}}
\newcommand{\linear}{\mathcal{L}}
\newcommand{\id}{\mathcal{I}}

\DeclareMathOperator{\tr}{tr}

\begin{document}

\title{AN UNCERTAINTY RELATION FOR MEASUREMENTS OF RANDOM UNITARY CHANNELS ACTING ON A QUBIT}
\author{Taihei Kimoto\footnote{kimoto.taihei.22e@st.kyoto-u.ac.jp} \ and Takayuki Miyadera\footnote{miyadera@nucleng.kyoto-u.ac.jp}}
\affil{Department of Nuclear Engineering, Kyoto University, Kyoto daigaku-katsura\\
Nishikyo-ku, Kyoto 615-8540,
Japan}
\date{}
\maketitle

\begin{abstract}
By preparing an input state and measuring an observable for the output state, we can measure a quantum channel.
Following the formulation given by Xiao \emph{et al}., we study an uncertainty relation for ancilla-free measurements of random unitary channels acting on a qubit.
We obtain an explicit formula and give a necessary and sufficient condition for this formula to be nontrivial.
\end{abstract}

\section{Introduction}
In contrast to classical theory, quantum theory is inherently stochastic; i.e., it predicts probability distributions instead of measurement results in experiments.
Furthermore, there exist nontrivial inequalities known as preparation uncertainty relations (PURs) that impose a restriction on the probability distributions.
Several PURs have been derived, and their differences mainly come from the measures of uncertainty used in their formulation.
For instance, the Robertson\cite{PhysRev.34.163}, Maassen-Uffink\cite{PhysRevLett.60.1103}, and Landau-Pollak\cite{PhysRevA.71.052325,PhysRevA.76.062108} inequalities employ the standard deviation, Shannon entropy and maximum value of the probability distribution as the measures of uncertainty, respectively.
Among them, PURs that use informational entropy as the measure of uncertainty, such as Maassen-Uffink's inequality, are called entropic uncertainty relations.
Because of its affinity for quantum information processing, the entropic uncertainty relations have been well studied over the past decades\cite{RevModPhys.89.015002}.
Additionally, in Refs. \citen{doi:10.1063/1.3614503,doi:10.1063/5.0017854,Takakura_2021} it is shown that the PURs are closely related to another uncertainty relation called uncertainty relation for joint measurement.

While many studies have been devoted to PURs, the preparation of quantum states is a special case of the quantum channel as a quantum state can be regarded as a quantum channel from a system with one-dimensional Hilbert space.
Therefore, it is natural to extend PURs so that they can treat quantum channels.
Recently, Xiao \emph{et al}.\cite{PhysRevResearch.3.023077} obtained uncertainty relations for all quantum channels in several forms using the process positive-operator-valued measure (PPOVM) formalism introduced by Ziman\cite{PhysRevA.77.062112}.
These formulas reduce to the conventional PURs when restricted to the preparation of quantum states.

In addition to the uncertainty relations for all quantum channels, uncertainty relations for a specific class of quantum channels are  meaningful because some quantum information processing employs specific quantum channels.
Another advantage of restricting the quantum channel in this way is that it provides a simple and specific evaluation formula.
Therefore, this paper aims to obtain such a formula, which is easy to evaluate for a specific class of quantum channels.
We choose random unitary channels\cite{heinosaari_ziman_2011} acting on a qubit as specific quantum channels that appear in some quantum information processing, such as the bit flip channel\cite{nielsen_chuang_2010} and the definition of a private quantum channel\cite{892142}.

The reminder of the paper is organized as follows. Section \ref{sec:preliminaries} presents the preliminaries. We obtain an explicit formula of the uncertainty relation for random unitary channels acting on a qubit in Section \ref{sec:f_and_r}. Finally, Section \ref{sec:c_and_d} presents the conclusion and discussions of the results.

\section{Preliminaries} \label{sec:preliminaries}
In quantum theory, a quantum state is described using a density operator, i.e., a linear operator $\rho$ satisfying $0\leq\rho$ and $\tr\rho=1$.
To obtain an outcome, we have to make a measurement.
A measurement maps a quantum state to a probability distribution.
The map is well described by an observable.
Each observable is represented by a collection of operators $\{E_{m}\}_{m\in\Omega}$ satisfying $0\leq E_{m}\leq I$ and $\sum_{m\in\Omega}E_{m}=I$.
Here $\Omega$ denotes the outcome set and $I$ denotes the identity operator.
Such a collection is called a positive-operator-valued measure (POVM) which gives $\tr\left[E_{m}\rho\right]$ as the probability of obtaining an outcome $m$ for a quantum state $\rho$.

To formulate an uncertainty relation for measurements of random unitary channels, we need to explain how to measure quantum channels.
As in the paper of Xiao \emph{et al}.,\cite{PhysRevResearch.3.023077} we employ the framework for channel measurements introduced by Ziman.\cite{PhysRevA.77.062112}
A quantum channel is described using a linear, completely positive, and trace preserving map $\Lambda:\linear(\hil_{in})\rightarrow\linear(\hil_{out})$, where $\hil_{in}$ and $\hil_{out}$ denote the input and output systems, respectively.
In this paper, all quantum channels are random unitary channels acting on a qubit, and hence, $\hil_{in}=\hil_{out}=\hil=\mathbb{C}^{2}$ and
\begin{equation}
\Lambda=\sum_{x}p_{x}\mathcal{U}_{x}, \label{eq:random_unitary}
\end{equation}
where $\{p_{x}\}_{x}$ is a probability distribution and $\{\mathcal{U}_{x}\}_{x}$ is a collection of unitary channels (i.e., $\mathcal{U}_{x}$ is represented as $\mathcal{U}_{x}(\rho)=U_{x}\rho U_{x}^{\dagger}$ with a unitary operator $U_{x}$). In the framework of channel measurements established by Ziman, a measurement of $\Lambda$ is performed as follows.
\begin{enumerate}
\item Prepare a state $\rho$ on a composite system $\hil_{anc}\otimes\hil$, where $\hil_{anc}$ is a reference system.
\item Apply $\Lambda$ to $\hil$.
\item Measure $\hil_{anc}\otimes\hil$ with a POVM $E=\{E_{m}\}_{m\in\Omega}$ on the composite system.
\end{enumerate}
It is obvious from these procedures that a channel measurement is characterized by $(\rho, E)$; thus, we identify a channel measurement with $(\rho, E)$. Ziman showed that the probability of obtaining an outcome $m$ by a channel measurement $(\rho, E)$ is given by
\begin{equation}
p_{m}(\Lambda)=\tr\left[\tilde{E}_{m}J_{\Lambda}\right], \label{eq:p_born}
\end{equation}
where $\tilde{E}_{m}$ is an operator acting on $\hil\otimes\hil$ determined by $\rho$ and $E_{m}$ and $J_{\Lambda}$ is defined as $J_{\Lambda}:=\left(\id_{in}\otimes\Lambda\right)(\ket{\psi'_{+}}\bra{\psi'_{+}})$ by an unnormalized maximally entangled state $\ket{\psi'_{+}}:=\sum^{d-1}_{i=0}\ket{i\otimes i}$.
Each operator $\tilde{E}_{m}$ and their collection $\tilde{E}=\{\tilde{E}_m\}_{m\in\Omega}$ are called a process-channel effect
and a process POVM (PPOVM) respectively.
In this paper, we consider only ancilla-free measurements, i.e., channel measurements specified with $(\xi\otimes\rho, \{I_{anc}\otimes E_{m}\}_{m\in\Omega})$. In this case, its process-channel effects become
\begin{equation}
\tilde{E}_{m}=\rho^{T}\otimes E_{m}, \label{eq:no_anc_ppovm}
\end{equation}
where $^T$ denotes transposition under the bases $\{\ket{i}\bra{j}\}^{d-1}_{i,j=0}$.\cite{PhysRevA.77.062112}

\section{Formulation and Results} \label{sec:f_and_r}
In a version of PURs called a PUR for probabilistic measurements of quantum states, one prepares a quantum state $\rho$ and measure with a POVM $A:=\{A_{m}\}_{m\in\Omega_{A}}$ with probability $r$ and measure with another POVM $B:=\{B_{n}\}_{n\in\Omega_{B}}$ with probability $1-r$.
We denote using $\mathbf{p}(\rho)=\{p_{m}(\rho)\}_{m\in\Omega_{A}}$ and $\mathbf{q}(\rho)=\{q_{n}(\rho)\}_{n\in\Omega_{B}}$ the probability distributions of the outcomes for $A$ and $B$, respectively.
We can assume $\Omega_{A}\cap\Omega_{B}=\emptyset$ in general.
In other words, we know which measurement is performed, and each measurement outcome is accompanied by an index to indicate the obseravable.
In this scenario, the PUR claims that there exists a nontrivial trade-off between $\mathbf{p}(\rho)$ and $\mathbf{q}(\rho)$.
For instance, let $A$ and $B$ be nondegenerate sharp observables described using $A_{m}=\ket{a_{m}}\bra{a_{m}}$ and $B_{n}=\ket{b_{n}}\bra{b_{n}}$, respectively.
Suppose that we measure $A$ and $B$ randomly.
The probability of obtaining an outcome of either $m$ or $n$ is computed as $(1/2)p_{m}(\rho)+(1/2)q_{n}(\rho)$, and it is bounded as
\begin{equation}
\frac{1}{2}p_{m}(\rho)+\frac{1}{2}q_{n}(\rho)
\leq\frac{1}{2}\left(1+ \left|\braket{a_{m}|b_{n}}\right|\right). \notag
\end{equation}
This inequality is equivalent to the so-called Landau-Pollak uncertainty relation.
If $\left|\braket{a_{m}|b_{n}}\right|\neq 1$, one can conclude that no state can attain $p_{m}(\rho)=q_{n}(\rho)=1$.
Generally, we consider a subset $\mathcal{M}\subset\Omega_{A}$ and $\mathcal{N}\subset\Omega_{B}$ and examine an upper bound of the value
\begin{equation}
r\sum_{m\in\mathcal{M}}p_{m}(\rho)+(1-r)\sum_{n\in\mathcal{N}}q_{n}(\rho), \label{eq:pur_prob}
\end{equation}
i.e., the probability of obtaining an outcome included in either $\mathcal{M}$ or $\mathcal{N}$. If the maximum of (\ref{eq:pur_prob}) satisfies 
\begin{align}
&\max_{\rho}\left(r\sum_{m\in\mathcal{M}}p_{m}(\rho)+(1-r)\sum_{n\in\mathcal{N}}q_{n}(\rho)\right) \notag \\
&<r\max_{\rho}\sum_{m\in\mathcal{M}}p_{m}(\rho)+(1-r)\max_{\rho}\sum_{n\in\mathcal{N}}q_{n}(\rho), \notag
\end{align}
we find a nontrivial trade-off. For more details, please refer to Ref. \citen{PhysRevA.89.052115}.

The scenario of our uncertainty relation for quantum channels is similar to this PUR.
For a quantum channel $\Lambda$, let us consider an experiment where an ancilla-free channel measurement $\mathcal{T}_{1}=(\xi\otimes\rho_{in},E=\{I_{anc}\otimes E_{m}\}_{m\in\Omega_{E}})$ is performed with probability $r$, and another ancilla-free channel measurement $\mathcal{T}_{2}=(\xi'\otimes\tau_{in},F=\{I_{anc}\otimes F_{n}\}_{n\in\Omega_{F}})$ is performed with probability $1-r$.
Let $p_{m}(\Lambda)$ and $q_{n}(\Lambda)$ be the probabilities of the outcomes for $\mathcal{T}_{1}$ and $\mathcal{T}_{2}$, respectively.
Thus, we replaced the state $\rho$ with the quantum channel $\Lambda$ and the POVMs with the ancilla-free channel measurements.

As in the scenario for PURs, we are interested in a quantity
\begin{equation}
r\sum_{m\in\mathcal{M}}p_{m}(\Lambda)+(1-r)\sum_{n\in\mathcal{N}}q_{n}(\Lambda), \label{eq:ur_func}
\end{equation}
where $\mathcal{M}$ and $\mathcal{N}$ are arbitrary subsets of $\Omega_{E}$ and $\Omega_{F}$, respectively.
Its maximum value is defined as
\begin{equation}
C(\mathcal{M},\mathcal{N}):=\max_{\Lambda}\left(r\sum_{m\in\mathcal{M}}p_{m}(\Lambda)+(1-r)\sum_{n\in\mathcal{N}}q_{n}(\Lambda)\right), \label{eq:bound_def}
\end{equation}
where the maximization is taken over all random unitary channels. If this value is strictly smaller than a trivial value
\begin{equation}
T(\mathcal{M},\mathcal{N}):=r\max_{\Lambda}\sum_{m\in\mathcal{M}}p_{m}(\Lambda)+(1-r)\max_{\Lambda}\sum_{n\in\mathcal{N}}q_{n}(\Lambda),
\end{equation}
we conclude that a nontrivial trade-off exists between two probability distributions.
The following theorem gives a concrete formula of (\ref{eq:bound_def}).
\begin{theorem} \label{theorem:main}
A random unitary channel $\Lambda$ is measured with $\mathcal{T}_{1}=(\xi\otimes\rho_{in},E=\{I_{anc}\otimes E_{m}\}_{m\in\Omega_{E}})$ with probability $r$ and $\mathcal{T}_{2}=(\xi'\otimes\tau_{in},F=\{I_{anc}\otimes F_{n}\}_{n\in\Omega_{F}})$ with probability $1-r$. Let $p_{m}(\Lambda)$ and $q_{n}(\Lambda)$ be the probabilities of $\mathcal{T}_{1}$ and $\mathcal{T}_{2}$, respectively. Then, the following inequality holds
\begin{equation}
r\sum_{m\in\mathcal{M}}p_{m}(\Lambda)+(1-r)\sum_{n\in\mathcal{N}}q_{n}(\Lambda)
\leq C(\mathcal{M},\mathcal{N}), \label{eq:main_result}
\end{equation}
where $\mathcal{M}$ and $\mathcal{N}$ are arbitrary subsets of $\Omega_{E}$ and $\Omega_{F}$, respectively.
Here, the right-hand side is given by
\begin{align}
&C(\mathcal{M},\mathcal{N}) \notag \\
&=\frac{Z(\mathcal{M},\mathcal{N})}{2}+\frac{1}{2}
\sqrt{
\begin{aligned}
&\left((r\Delta\rho_{in}\Delta E(\mathcal{M})-(1-r)\Delta\tau_{in}\Delta F(\mathcal{N})\right)^{2} \\
&+4r(1-r)\Delta\rho_{in}\Delta E(\mathcal{M})\Delta\tau_{in}\Delta F(\mathcal{N})\Biggl(\left|\braket{\psi|\phi}\right|\left|\braket{\psi'|\phi'}\right| \\
&+\sqrt{\left(1-\left|\braket{\psi|\phi}\right|^{2}\right)\left(1-\left|\braket{\psi'|\phi'}\right|^{2}\right)}\Biggr)^{2}
\end{aligned}
}, \label{eq:ur_bound}
\end{align}
where $E(\mathcal{M}):=\sum_{m\in\mathcal{M}}E_{m}$ and $F(\mathcal{N}):=\sum_{n\in\mathcal{N}}F_{n}$ and $\Delta(\cdot)$ denotes the difference between the maximum and minimum eigenvalues of its argument and $\ket{\psi},\ket{\phi},\ket{\psi'},\ket{\phi'}$ are eigenvectors of the maximum eigenvalue of $\rho_{in},\tau_{in},E(\mathcal{M}),F(\mathcal{N})$, respectively.
\end{theorem}
\begin{proof}
Obviously, an inequality
\begin{equation}
\max_{\mathcal{U}}\left(r\sum_{m\in\mathcal{M}}p_{m}(\mathcal{U})+(1-r)\sum_{n\in\mathcal{N}}q_{n}(\mathcal{U})\right)\leq C(\mathcal{M},\mathcal{N}) \label{eq:c_is_larger_than_u}
\end{equation}
holds, where the optimization of the left-hand side is taken over all unitary channels. By using (\ref{eq:random_unitary}) and (\ref{eq:p_born}), we get the reverse inequality
\begin{align}
&r\sum_{m\in\mathcal{M}}p_{m}(\Lambda)+(1-r)\sum_{n\in\mathcal{N}}q_{n}(\Lambda) \notag \\
&=\sum_{x}p_{x}\tr\left[\left(r\sum_{m\in\mathcal{M}}\tilde{E}_{m}+(1-r)\sum_{n\in\mathcal{N}}\tilde{F}_{n}\right)J_{\mathcal{U}_{x}}\right] \notag \\
&\leq\sum_{x}p_{x}\max_{\mathcal{U}}\tr\left[\left(r\sum_{m\in\mathcal{M}}\tilde{E}_{m}+(1-r)\sum_{n\in\mathcal{N}}\tilde{F}_{n}\right)J_{\mathcal{U}}\right] \notag \\
&=\max_{\mathcal{U}}\left(r\sum_{m\in\mathcal{M}}p_{m}(\mathcal{U})+(1-r)\sum_{n\in\mathcal{N}}q_{n}(\mathcal{U})\right), \notag
\end{align}
for arbitrary random unitary channels $\Lambda$. Therefore, (\ref{eq:c_is_larger_than_u}) becomes equality; i.e.,
\begin{align}
C(\mathcal{M},\mathcal{N})
&=\max_{\mathcal{U}}\left(r\sum_{m\in\mathcal{M}}p_{m}(\mathcal{U})+(1-r)\sum_{n\in\mathcal{N}}q_{n}(\mathcal{U})\right) \notag \\
&=\max_{\mathcal{U}}\tr\left[\left(r\sum_{m\in\mathcal{M}}\tilde{E}_{m}+(1-r)\sum_{n\in\mathcal{N}}\tilde{F}_{n}\right)J_{\mathcal{U}}\right]. \label{eq:ur_func_est}
\end{align}
We introduce a quantity called fully entangled fraction (FEF),\cite{PhysRevA.54.3824} which is closely related to the right-hand side of (\ref{eq:ur_func_est}).
For a state $\rho$ on the composite system $\hil\otimes\hil$, the FEF is defined as
\begin{equation}
f(\rho)=\max_{U}\bra{\psi_{+}}\left(I\otimes U^{\dagger}\right)\rho\left(I\otimes U\right)\ket{\psi_{+}}, \label{eq:fef_def}
\end{equation}
where the maximization in the right-hand side is taken over all unitary operators, and $\ket{\psi_{+}}$ is a maximally entangled state defined as $\ket{\psi_{+}}:=\ket{\psi'_{+}}/\sqrt{2}$. By substituting the definition of $J_{\mathcal{U}}$ into the right-hand side of (\ref{eq:ur_func_est}), we obtain
\begin{equation}
C(\mathcal{M},\mathcal{N})
=2\max_{U}\bra{\psi_{+}}\left(I\otimes U^{\dagger}\right)\left(r\sum_{m\in\mathcal{M}}\tilde{E}_{m}+(1-r)\sum_{n\in\mathcal{N}}\tilde{F}_{n}\right)\left(I\otimes U\right)\ket{\psi_{+}}. \label{eq:ur_func_est_before_fef}
\end{equation}
Furthermore, we can define a state related to the process-channel effect as
\begin{equation}
\rho(\mathcal{M},\mathcal{N}):=\frac{r\sum_{m\in\mathcal{M}}\tilde{E}_{m}+(1-r)\sum_{n\in\mathcal{N}}\tilde{F}_{n}}{Z(\mathcal{M},\mathcal{N})},
\end{equation}
where $Z(\mathcal{M},\mathcal{N}):=\tr\left[r\sum_{m\in\mathcal{M}}\tilde{E}_{m}+(1-r)\sum_{n\in\mathcal{N}}\tilde{F}_{n}\right]$. Hence, by using this state and (\ref{eq:fef_def}), we obtain
\begin{equation}
C(\mathcal{M},\mathcal{N})=2Z(\mathcal{M},\mathcal{N})f\left(\rho(\mathcal{M},\mathcal{N})\right). \label{eq:ur}
\end{equation}
Therefore, the problem of calculating $C(\mathcal{M},\mathcal{N})$ is translated into the problem of calculating the corresponding value of the FEF.

The process-channel effects of $\mathcal{T}_{1}$ and $\mathcal{T}_{2}$ are $\tilde{E}_{m}=\rho^{T}_{in}\otimes E_{m}$ and $\tilde{F}_{n}=\tau^{T}_{in}\otimes F_{n}$, respectively. Thus, we obtain $Z(\mathcal{M},\mathcal{N})=\tr\left[rE(\mathcal{M})+(1-r)F(\mathcal{N})\right]$ and the state
\begin{equation}
\rho(\mathcal{M},\mathcal{N})
=R\rho^{T}_{in}\otimes\frac{E(\mathcal{M})}{\tr E(\mathcal{M})}+(1-R)\tau^{T}_{in}\otimes\frac{F(\mathcal{N})}{\tr F(\mathcal{N})}, \notag
\end{equation}
where $E(\mathcal{M}):=\sum_{m\in\mathcal{M}}E_{m},F(\mathcal{N}):=\sum_{n\in\mathcal{N}}F_{n}$ and $R:=\tr\left[rE(\mathcal{M})\right]/Z(\mathcal{M},\mathcal{N})$. By applying Lemma \ref{lemma:fef_formula} in the Appendix, we obtain
\begin{align}
&f(\rho(\mathcal{M},\mathcal{N})) \notag \\
&=\frac{1}{4}+\frac{1}{4Z(\mathcal{M},\mathcal{N})}
\sqrt{
\begin{aligned}
&\left((r\Delta\rho_{in}\Delta E(\mathcal{M})-(1-r)\Delta\tau_{in}\Delta F(\mathcal{N})\right)^{2} \\
&+4r(1-r)\Delta\rho_{in}\Delta E(\mathcal{M})\Delta\tau_{in}\Delta F(\mathcal{N})\Biggl(\left|\braket{\psi|\phi}\right|\left|\braket{\psi'|\phi'}\right| \\
&+\sqrt{\left(1-\left|\braket{\psi|\phi}\right|^{2}\right)\left(1-\left|\braket{\psi'|\phi'}\right|^{2}\right)}\Biggr)^{2}
\end{aligned}
} \notag
\end{align}
where $\ket{\psi}$,$\ket{\phi}$,$\ket{\psi'}$, and $\ket{\phi'}$ are the eigenvectors of the maximum eigenvalue of $\rho_{in}$, $\tau_{in}$, $E(\mathcal{M})$, and $F(\mathcal{N})$, respectively. Finally, by substituting this value into the right-hand side of (\ref{eq:ur})
\begin{align}
&C(\mathcal{M},\mathcal{N}) \notag \\
&=\frac{Z(\mathcal{M},\mathcal{N})}{2}+\frac{1}{2}
\sqrt{
\begin{aligned}
&\left((r\Delta\rho_{in}\Delta E(\mathcal{M})-(1-r)\Delta\tau_{in}\Delta F(\mathcal{N})\right)^{2} \\
&+4r(1-r)\Delta\rho_{in}\Delta E(\mathcal{M})\Delta\tau_{in}\Delta F(\mathcal{N})\Biggl(\left|\braket{\psi|\phi}\right|\left|\braket{\psi'|\phi'}\right| \\
&+\sqrt{\left(1-\left|\braket{\psi|\phi}\right|^{2}\right)\left(1-\left|\braket{\psi'|\phi'}\right|^{2}\right)}\Biggr)^{2}
\end{aligned}
}. \notag
\end{align}
\end{proof}
Let us compare $C(\mathcal{M},\mathcal{N})$ with the trivial upper bound $T(\mathcal{M},\mathcal{N})$. By using Lemma \ref{lemma:fef_formula_prod}, we obtain the value of $T(\mathcal{M},\mathcal{N})$ as
\begin{equation}
T(\mathcal{M},\mathcal{N})=\frac{1}{2}\left(r\left(\tr E(\mathcal{M})+\Delta\rho_{in}\Delta E(\mathcal{M})\right)+(1-r)\left(\tr F(\mathcal{N})+\Delta\tau_{in}\Delta F(\mathcal{N})\right)\right). \label{eq:ur_bound_trivial}
\end{equation}
It is not difficult to see that $C(\mathcal{M},\mathcal{N})=T(\mathcal{M},\mathcal{N})$ holds for $0<r<1$ if and only if $\left|\braket{\psi|\phi}\right|=\left|\braket{\psi'|\phi'}\right|$ holds.
Therefore we obtain the following observation.
\begin{proposition}
Consider the same situation as in Theorem 1.
The quantity $C(\mathcal{M},\mathcal{N})$ gives a nontrivial bound for $0<r<1$ in a sense that $C(\mathcal{M},\mathcal{N})<T(\mathcal{M},\mathcal{N})$ if and only if $\left|\braket{\psi|\phi}\right|\neq\left|\braket{\psi'|\phi'}\right|$.
\end{proposition}

To provide an example of Theorem \ref{theorem:main}, we consider an ideal case; i.e., pure states are input, and measurements are performed with rank-one projection-valued measures (PVMs).
\begin{corollary} \label{cor:sec_result}
Let $\Omega_{E}:=\{m_{0},m_{1}\}$ and $\Omega_{F}:=\{n_{0},n_{1}\}$, and $\{\ket{e_{m}}\}_{m\in\Omega_{E}}$ and $\{\ket{f_{n}}\}_{n\in\Omega_{F}}$ be arbitrary orthonormal bases, and $\ket{\psi_{in}}$ and $\ket{\phi_{in}}$ be arbitrary pure states. A random unitary channel $\Lambda$ is measured with $\mathcal{T}_{1}=(\xi\otimes\ket{\psi_{in}}\bra{\psi_{in}},E=\{I_{anc}\otimes \ket{e_{m}}\bra{e_{m}}\}_{m\in\Omega_{E}})$ with probability $r$ and $\mathcal{T}_{2}=(\xi'\otimes\ket{\phi_{in}}\bra{\phi_{in}},F=\{I_{anc}\otimes \ket{f_{n}}\bra{f_{n}}\}_{n\in\Omega_{F}})$ with probability $1-r$. Then, the following inequality holds
\begin{equation}
rp_{m}(\Lambda)+(1-r)q_{n}(\Lambda)
\leq C(\{m\},\{n\}), \label{eq:sec_result}
\end{equation}
where $p_{m}(\Lambda)$ and $q_{n}(\Lambda)$ are probabilities of obtaining outcomes $m$ and $n$ for $\mathcal{T}_{1}$ and $\mathcal{T}_{2}$, respectively, and $C(\{m\},\{n\})$ is given by
\begin{equation}
C(\{m\},\{n\})
=\frac{1}{2}+\sqrt{
\begin{aligned}
&\left(r-\frac{1}{2}\right)^{2}+r(1-r)\Biggl(\left|\braket{\psi_{in}|\phi_{in}}\right|\left|\braket{e_{m}|f_{n}}\right| \\
&+\sqrt{\left(1-\left|\braket{\psi_{in}|\phi_{in}}\right|^{2}\right)\left(1-\left|\braket{e_{m}|f_{n}}\right|^{2}\right)}\Biggr)^{2}
\end{aligned}
}.
\label{eq:sec_result_bound}
\end{equation}
\end{corollary}
\begin{proof}
By calculating directly, $Z(\mathcal{M},\mathcal{N})=1$. Because $\Delta(\cdot)$ means the difference between the maximum and the minimum eigenvalues of the argument, for pure states, their values are unity. Obviously, $\ket{\psi}=\ket{\psi_{in}},\ket{\phi}=\ket{\phi_{in}},\ket{\psi'}=\ket{e_{m}},\ket{\phi'}=\ket{f_{n}}$. From these facts, (\ref{eq:ur_bound}) becomes
\begin{equation}
C(\{m\},\{n\})
=\frac{1}{2}+\sqrt{
\begin{aligned}
&\left(r-\frac{1}{2}\right)^{2}+r(1-r)\Biggl(\left|\braket{\psi_{in}|\phi_{in}}\right|\left|\braket{e_{m}|f_{n}}\right| \\
&+\sqrt{\left(1-\left|\braket{\psi_{in}|\phi_{in}}\right|^{2}\right)\left(1-\left|\braket{e_{m}|f_{n}}\right|^{2}\right)}\Biggr)^{2}
\end{aligned}
} \notag
\end{equation}
\end{proof}
The value of $T(\{m\},\{n\})$ is easily calculated as $T(\{m\},\{n\})=1$ from (\ref{eq:ur_bound_trivial}). Now, let us consider the following example.
\begin{example} \label{ex:evidence}
In Corollary \ref{cor:sec_result}, consider a case where the input states are $\ket{\psi_{in}}=\ket{0}$ and $\ket{\phi_{in}}=\cos(\theta/2)\ket{0}+\sin(\theta/2)\ket{1}$ with $0\leq\theta\leq\pi$, and the measurements for the output states are performed with $\{\ket{e_{m_{0}}}=\ket{0},\ket{e_{m_{1}}}=\ket{1}\}$ and $\{\ket{f_{n_{0}}}=\ket{x_{0}}:=(\ket{0}+\ket{1})/\sqrt{2},\ket{f_{n_{1}}}=\ket{x_{1}}:=(\ket{0}-\ket{1})/\sqrt{2}\}$, and $r=1/2$.
The formula (\ref{eq:sec_result_bound}) gives
\begin{equation}
C(\{m\},\{n\})=\frac{1}{2}+\frac{1}{2}\cos\left(\frac{\theta}{2}-\frac{\pi}{4}\right). \label{eq:ex_bound}
\end{equation}
Therefore, the equality $C(\{m\},\{n\})=T(\{m\},\{n\})$ holds if and only if $\theta=\pi/2$.
\end{example}

\section{Conclusions and Discussions} \label{sec:c_and_d}
In this paper, we investigated an uncertainty relation specialized in random unitary channels acting on a qubit.
Theorem \ref{theorem:main} presents an explicit formula for a pair of ancilla-free channel measurements.
Although our inequality can be applied to rather limited situations, the bound is so simple that it can easily be calculated.
Such an explicit bound is expected to be useful in drawing any physical observations.

It is interesting to compare our bound with those derived by Xiao \emph{et al}.
Let us consider Example \ref{ex:evidence} with $\theta=\pi$.
In this case, the bound (\ref{eq:ex_bound}) becomes
\begin{equation}
C(\{m\},\{n\})=\frac{1}{2}\left(1+\frac{1}{\sqrt{2}}\right)<1. \label{eq:example_for_violation}
\end{equation}
This means that there are no random unitary channels so that both probability distributions $\mathbf{p}(\Lambda):=\{p_{m}(\Lambda)\}_{m\in\Omega_{E}}$ and $\mathbf{q}(\Lambda):=\{q_{n}(\Lambda)\}_{n\in\Omega_{F}}$ become sharp.
However, a measure-and-prepare channel defined as
\begin{equation}
\Phi(\rho):=\ket{0}\bra{0}\rho\ket{0}\bra{0}+\ket{x_{0}}\bra{1}\rho\ket{1}\bra{x_{0}}.
\end{equation}
attains $p_{m_{0}}(\Phi)=q_{n_{0}}(\Phi)=1$. This fact shows that the nontriviality of (\ref{eq:example_for_violation}) is due to the restriction of the channels to random unitary ones.

Finally, we mention future problems related to our uncertainty relation.
Our experimental scenario is not the most general one; i.e., the channel measurements are limited to ancilla-free, and the evolving system is a qubit.
In the proof of Theorem \ref{theorem:main}, we showed that the bound of our uncertainty relation is given by the FEF (see (\ref{eq:ur})).
However, an analytical solution for the FEF of arbitrary quantum states is only known in two-qubit systems.
To generalize our results, one may be able to use an upper bound of the FEF derived in Ref. \citen{PhysRevA.78.032332}, but it is not trivial. Another direction of study is the application of our uncertainty relation to quantum information processing. This direction of study is also nontrivial, but the application of the entropic uncertainty relations to quantum information processing may be helpful\cite{RevModPhys.89.015002}.

\section*{Acknowledgments}
This work was supported by JST SPRING, Grant Number JPMJSP2110. TM acknowledges financial support from JSPS (KAKENHI Grant Number
20K03732).

\section*{Appendix A. The FEF of separable states}
In this section, we provide a formula for the FEF of separable states consisting of two product states. The FEF was first introduced in Ref. \citen{PhysRevA.54.3824}, and its analytical solution for two-qubit systems was derived in Ref. \citen{GRONDALSKI2002573}. In the course of obtaining an upper bound for two-qudit systems, the analytical solution was translated into the following theorem.\cite{PhysRevA.78.032332}
\begin{theorem} \label{theorem:fef_qubit}
Let $\{\lambda_{i}\}^{3}_{i=1}$ and $\{\mu_{j}\}^{3}_{j=1}$ be arbitrary generators of $SU(2)$ with the orthogonality relations $\tr\lambda_{i}\lambda_{j}=2\delta_{ij}$ and $\tr\mu_{i}\mu_{j}=2\delta_{ij}$. For an arbitrary two-qubit state $\rho$, the FEF is given by
\begin{equation}
f(\rho)=\frac{1}{4}+4\|M(\rho)^{T}M(P_{+})\|_{KF}, \label{eq:fef_qubit}
\end{equation}
where $M(\cdot)$ denotes a matrix whose elements are given by $(M(\cdot))_{ij}=(1/4)\tr\left[(\cdot)(\lambda_{i}\otimes\mu_{j})\right]$, and $P_{+}:=\ket{\psi_{+}}\bra{\psi_{+}}$, and $M(\cdot)^{T}$ denotes transposed matrix, and $\|\cdot\|_{KF}$ is the Ky Fan norm defined as $\|X\|_{KF}:=\tr\sqrt{XX^{\dagger}}$ for an arbitrary matrix $X$.
\end{theorem}

By using Theorem \ref{theorem:fef_qubit}, we obtain following lemma.
\begin{lemma} \label{lemma:fef_formula}
Suppose that a two-qubit state $\rho$ is represented by
\begin{equation}
\rho=r\rho_{A}\otimes\rho_{B}+(1-r)\tau_{A}\otimes\tau_{B}, \notag
\end{equation}
where $0\leq r\leq 1$. The FEF of the state $\rho$ is given by
\begin{equation}
f(\rho)=\frac{1}{4}+\frac{1}{4}
\sqrt{
\begin{aligned}
&\left((r\Delta\rho_{A}\Delta\rho_{B})-(1-r)\Delta\tau_{A}\Delta\tau_{B}\right)^{2} \\
&+4r(1-r)\Delta\rho_{A}\Delta\rho_{B}\Delta\tau_{A}\Delta\tau_{B}\Biggl(\left|\braket{\psi_{A}|\phi_{A}}\right|\left|\braket{\psi_{B}|\phi_{B}}\right| \\
&+\sqrt{\left(1-\left|\braket{\psi_{A}|\phi_{A}}\right|^{2}\right)\left(1-\left|\braket{\psi_{B}|\phi_{B}}\right|^{2}\right)}\Biggr)^{2}
\end{aligned}
},
\label{eq:fef_formula}
\end{equation}
where $\Delta(\cdot)$ denotes the difference between the maximum and the minimum eigenvalues of its argument and $\ket{\psi_{X}}$ and $\ket{\phi_{X}}$ denote the eigenvectors of the maximum eigenvalues of the $\rho_{X}$ and $\tau_{X}$, respectively ($X=A,B$).
\end{lemma}
\begin{proof}
Let $\ket{\psi_{X,0}}$ and $\ket{\phi_{X,0}}$ be the eigenvectors of the maximal eigenvalues of $\rho_{X}$ and $\tau_{X}$, respectively ($X=A, B$). The minimal eigenvalues of $\rho_{X}$ and $\tau_{X}$ are $1-\|\rho_{X}\|$ and $1-\|\tau_{X}\|$, and let $\ket{\psi_{X,1}}$ and $\ket{\phi_{X,1}}$ be the corresponding eigenvectors, respectively. The spectral decompositions of $\rho_{X}$ and $\tau_{X}$ are given by
\begin{align}
\rho_{X}&=\|\rho_{X}\|\ket{\psi_{X,0}}\bra{\psi_{X,0}}+(1-\|\rho_{X}\|)\|\ket{\psi_{X,1}}\bra{\psi_{X,1}}, \notag \\
\tau_{X}&=\|\tau_{X}\|\ket{\phi_{X,0}}\bra{\phi_{X,0}}+(1-\|\tau_{X}\|)\|\ket{\phi_{X,1}}\bra{\phi_{X,1}}. \notag
\end{align}
By substituting the completeness conditions $\ket{\psi_{X,1}}\bra{\psi_{X,1}}=I-\ket{\psi_{X,0}}\bra{\psi_{X,0}}$ and $\ket{\phi_{X,1}}\bra{\phi_{X,1}}=I-\ket{\phi_{X,0}}\bra{\phi_{X,0}}$, we obtain
\begin{align}
\rho_{X}&=(2\|\rho_{X}\|-1)\ket{\psi_{X,0}}\bra{\psi_{X,0}}+(1-\|\rho_{X}\|)I, \notag \\
\tau_{X}&=(2\|\tau_{X}\|-1)\ket{\phi_{X,0}}\bra{\phi_{X,0}}+(1-\|\tau_{X}\|)I. \notag
\end{align}
The FEF is invariant under the local unitary transformations, i.e., for an arbitrary state $\rho$, the equality
\begin{equation}
f(\rho)=f((V_{A}\otimes V_{B})\rho(V_{A}\otimes V_{B})^{\dagger}) \notag
\end{equation}
follows, where $V_{A}$ and $V_{B}$ are unitary transformations. This property is apparent from $V^{\dagger}_{A}\otimes V^{\dagger}_{B}\ket{\psi_{+}}=I\otimes V^{\dagger}_{B}(V^{\dagger}_{A})^{T}\ket{\psi_{+}}$ where $^{T}$ denotes transposition in the $\{\ket{m}\bra{n}\}^{1}_{m,n=0}$.
Let $\varphi_{X,m}:=-\arg\braket{\psi_{X,m}|\phi_{X,0}}$, we define local unitary transformations as
\begin{equation}
V_{X}:=e^{i\varphi_{X,0}}\ket{0}\bra{\psi_{X,0}}+e^{i\varphi_{X,1}}\ket{1}\bra{\psi_{X,1}}. \notag
\end{equation}
From
\begin{align}
V_{X}\ket{\psi_{X,0}}&=e^{i\varphi_{X,0}}\ket{0}, \notag \\
V_{X}\ket{\phi_{X,0}}
&=|\braket{\psi_{X,0}|\phi_{X,0}}|\ket{0}+|\braket{\psi_{X,1}|\phi_{X,0}}|\ket{1}
=:\ket{\phi_{X,0}'}, \notag
\end{align}
$V_{X}$ operates as
\begin{align}
V_{X}\rho_{X}V^{\dagger}_{X}
&=(2\|\rho_{X}\|-1)\ket{0}\bra{0}+(1-\|\rho_{X}\|)I
=:\rho_{X}', \notag \\
V_{X}\tau_{X}V^{\dagger}_{X}
&=(2\|\tau_{X}\|-1)\ket{\phi_{X,0}'}\bra{\phi_{X,0}'}+(1-\|\tau_{X}\|)I
=:\tau_{X}'. \notag
\end{align}
Therefore, $V_{A}\otimes V_{B}$ operates as
\begin{equation}
(V_{A}\otimes V_{B})\rho(V^{\dagger}_{A}\otimes V^{\dagger}_{B})
=r\rho_{A}'\otimes\rho_{B}'+(1-r)\tau_{A}'\otimes\tau_{B}'
=:\rho'. \notag
\end{equation}
Because the FEF is invariant under local unitary transformations, we calculate the FEF of $\rho'$ instead of $\rho$, and its value is given by (\ref{eq:fef_qubit}), i.e.,
\begin{equation}
f(\rho')=\frac{1}{4}+4\|M(\rho')^{T}M(P_{+})\|_{KF}. \notag
\end{equation}
We choose the Pauli operators
\begin{align}
\sigma_{1}&:=\ket{0}\bra{1}+\ket{1}\bra{0}, \notag \\
\sigma_{2}&:=(-i)\ket{0}\bra{1}+i\ket{1}\bra{0}, \notag \\
\sigma_{3}&:=\ket{0}\bra{0}-\ket{1}\bra{1}, \notag
\end{align}
as the generators, which give the elements of $M(\cdot)$. By calculating directly, we obtain
\begin{equation}
M(P_{+})M(P_{+})^{T}=\frac{1}{16}I. \notag
\end{equation}
Therefore, the FEF is given by
\begin{equation}
f(\rho')=\frac{1}{4}+\|M(\rho')^{T}\|_{KF}. \notag
\end{equation}
By the definition of $M(\cdot)$, the elements of $M(\rho')$ is given by
\begin{equation}
(M(\rho'))_{ij}
=\frac{1}{4}\left(r\left(\tr\rho_{A}'\sigma_{i}\right)\left(\tr\rho_{B}'\sigma_{j}\right)+(1-r)\left(\tr\tau_{A}'\sigma_{i}\right)\left(\tr\tau_{B}'\sigma_{j}\right)\right). \label{eq:corr_rho}
\end{equation}
Let $b_{X,i}:=\braket{\phi_{X,0}'|\sigma_{i}|\phi_{X,0}'}$, and by substituting
\begin{align}
\tr\rho_{X}'\sigma_{i}&=(2\|\rho_{X}\|-1)\delta_{i3}, \notag \\
\tr\tau_{X}'\sigma_{i}&=(2\|\tau_{X}\|-1)b_{X,i}(1-\delta_{i2}), \notag
\end{align}
into (\ref{eq:corr_rho}), we obtain
\begin{equation}
(M(\rho'))_{ij}
=\frac{1}{4}\left(r\Delta\rho_{A}\Delta\rho_{B}\delta_{i3}\delta_{j3}+(1-r)\Delta\tau_{A}\Delta\tau_{B}b_{A,i}b_{B,j}(1-\delta_{i2})(1-\delta_{j2})\right), \notag
\end{equation}
i.e., its matrix is
\begin{equation}
M(\rho')=
\frac{1}{4}
\begin{pmatrix}
(1-r)\Delta\tau_{A}\Delta\tau_{B}b_{A,1}b_{B,1} & 0 & (1-r)\Delta\tau_{A}\Delta\tau_{B}b_{A,1}b_{B,3} \\
0 & 0 & 0 \\
(1-r)\Delta\tau_{A}\Delta\tau_{B}b_{A,3}b_{B,1} & 0 & r\Delta\rho_{A}\Delta\rho_{B}+(1-r)\Delta\tau_{A}\Delta\tau_{B}b_{A,3}b_{B,3} \\
\end{pmatrix}
. \notag
\end{equation}
From the definition of Ky Fan norm, we need to calculate $M(\rho')^{T}M(\rho')$.
However, because the values of the second column and the second row of $M(\rho')$ are all zero, it is sufficient to calculate  $M(\rho'))_{1,1},M(\rho'))_{1,3},M(\rho'))_{3,1},M(\rho'))_{3,3}$.
Additionally, because $M(\rho')^{T}M(\rho')$ is symmetric, $M(\rho'))_{1,3}=M(\rho'))_{3,1}$. By calculating directly, we obtain
\begin{align}
(M(\rho')^{T}M(\rho'))_{11}&=\frac{1}{16}\left((1-r)\Delta\tau_{A}\Delta\tau_{B}b_{B,1}\right)^{2}, \notag \\
(M(\rho')^{T}M(\rho'))_{13}&=\frac{1}{16}(1-r)\Delta\tau_{A}\Delta\tau_{B}b_{B,1}((1-r)\Delta\tau_{A}\Delta\tau_{B}b_{B,3} \notag \\
&\quad+r\Delta\rho_{A}\Delta\rho_{B}b_{A,3}), \notag
\end{align}
and
\begin{align}
(M(\rho')^{T}M(\rho'))_{33}
&=\frac{1}{16}\bigl(\left((1-r)\Delta\tau_{A}\Delta\tau_{B}b_{B,3}\right)^{2}+(r\Delta\rho_{A}\Delta\rho_{B})^{2} \notag \\
&\quad+2r(1-r)\Delta\rho_{A}\Delta\rho_{B}\Delta\tau_{A}\Delta\tau_{B}b_{A,3}b_{B,3}\bigr). \notag
\end{align}
Therefore, by defining
\begin{align}
x&:=(M(\rho')^{T}M(\rho'))_{11}, \notag \\
y&:=(M(\rho')^{T}M(\rho'))_{13}=(M(\rho')^{T}M(\rho'))_{31}, \notag \\
z&:=(M(\rho')^{T}M(\rho'))_{33}, \label{eq:x_y_z_def}
\end{align}
$M(\rho')^{T}M(\rho')$ is represented as
\begin{equation}
M(\rho')^{T}M(\rho')=
\begin{pmatrix}
x & 0 & y \\
0 & 0 & 0 \\
y & 0 & z \\
\end{pmatrix}
. \notag \\
\end{equation}
In this case, the Ky Fan norm of $M(\rho')^{T}$ is given by the formula
\begin{equation}
\|M(\rho')^{T}\|_{KF}=\sqrt{x+z+2\sqrt{xz-y^{2}}}. \label{eq:corr_rho_formula}
\end{equation}
The formula (\ref{eq:corr_rho_formula}) can be obtained from the identity
\begin{equation}
\|M(\rho)^{T}\|_{KF}=\sqrt{(\|M(\rho)^{T}\|_{KF})^{2}}=\sqrt{\left(\sqrt{\lambda_{+}}+\sqrt{\lambda_{-}}\right)^{2}}=\sqrt{\lambda_{+}+\lambda_{-}+2\sqrt{\lambda_{+}\lambda_{-}}}, \notag
\end{equation}
where $\lambda_{\pm}$ are eigenvalues of $M(\rho')^{T}M(\rho')$, and are defined as
\begin{equation}
\lambda_{\pm}:=\frac{x+z\pm\sqrt{(x+z)^{2}-4(xz-y^{2})}}{2}. \notag
\end{equation}
By substituting (\ref{eq:x_y_z_def}) into (\ref{eq:corr_rho_formula}), we obtain
\begin{equation}
\|M(\rho')^{T}\|_{KF}
=\frac{1}{4}\sqrt{
\begin{aligned}
&\left((r\Delta\rho_{A}\Delta\rho_{B})-(1-r)\Delta\tau_{A}\Delta\tau_{B}\right)^{2} \\
&+4r(1-r)\Delta\rho_{A}\Delta\rho_{B}\Delta\tau_{A}\Delta\tau_{B}\left|\braket{\phi_{A,0}'|\phi_{B,0}'}\right|^{2}
\end{aligned}
}. \label{eq:corr_rho_prime}
\end{equation}
From the definition of $\ket{\phi_{X,0}'}$, the equation
\begin{align}
&\braket{\phi_{A,0}'|\phi_{B,0}'}= \notag \\
&\left|\braket{\psi_{A,0}|\phi_{A,0}}\right|\left|\braket{\psi_{B,0}|\phi_{B,0}}\right|+\sqrt{\left(1-\left|\braket{\psi_{A,0}|\phi_{A,0}}\right|^{2}\right)\left(1-\left|\braket{\psi_{B,0}|\phi_{B,0}}\right|^{2}\right)} \label{eq:prime_to_no_prime}
\end{align}
follows. Finally, by substituting (\ref{eq:prime_to_no_prime}) into (\ref{eq:corr_rho_prime}), and defining $\ket{\psi_{X}}:=\ket{\psi_{X,0}}$ and $\ket{\phi_{X}}:=\ket{\phi_{X,0}}$, we obtain (\ref{eq:fef_formula}).
\end{proof}

In Lemma \ref{lemma:fef_formula}, by substituting $r=1$, we can get a formula for the FEF of product states as follows.
\begin{lemma} \label{lemma:fef_formula_prod}
The FEF of a state $\rho:=\rho_{A}\otimes\rho_{B}$ is given by
\begin{equation}
f(\rho)=\frac{1}{4}+\frac{1}{4}\Delta\rho_{A}\Delta\rho_{B}.
\end{equation}
\end{lemma}



\begin{thebibliography}{10}

\bibitem{PhysRev.34.163}
H.~P. Robertson, {\em Phys. Rev.} {\bf 34} (1929) 163.

\bibitem{PhysRevLett.60.1103}
H.~Maassen and J.~B.~M. Uffink, {\em Phys. Rev. Lett.} {\bf 60} (1988)
  1103.

\bibitem{PhysRevA.71.052325}
J.~I. de~Vicente and J.~S\'anchez-Ruiz, {\em Phys. Rev. A} {\bf 71} (2005)
  052325.

\bibitem{PhysRevA.76.062108}
T.~Miyadera and H.~Imai, {\em Phys. Rev. A} {\bf 76} (2007) 062108.

\bibitem{RevModPhys.89.015002}
P.~J. Coles, M.~Berta, M.~Tomamichel and S.~Wehner, {\em Rev. Mod. Phys.} {\bf
  89} (2017) 015002.

\bibitem{doi:10.1063/1.3614503}
T.~Miyadera, {\em J. Math. Phys.} {\bf 52}  (2011) 072105.

\bibitem{doi:10.1063/5.0017854}
R.~Takakura and T.~Miyadera, {\em J. Math. Phys.} {\bf 61}
  (2020) 082203.

\bibitem{Takakura_2021}
R.~Takakura and T.~Miyadera, {\em J. Phys. A Math. Theor.} {\bf 54} (2021) 315302.

\bibitem{PhysRevResearch.3.023077}
Y.~Xiao, K.~Sengupta, S.~Yang and G.~Gour, {\em Phys. Rev. Research} {\bf 3}
  (2021) 023077.

\bibitem{PhysRevA.77.062112}
M.~Ziman, {\em Phys. Rev. A} {\bf 77} (2008) 062112.

\bibitem{heinosaari_ziman_2011}
T.~Heinosaari and M.~Ziman, {\em The Mathematical Language of Quantum Theory:
  From Uncertainty to Entanglement} (Cambridge University Press, 2011).

\bibitem{nielsen_chuang_2010}
M.~A. Nielsen and I.~L. Chuang, {\em Quantum Computation and Quantum
  Information: 10th Anniversary Edition} (Cambridge University Press, 2010).

\bibitem{892142}
A.~Ambainis, M.~Mosca, A.~Tapp and R.~De~Wolf, Private quantum channels, in
  {\em Proceedings 41st Annual Symposium on Foundations of Computer Science} (2000), pp. 547–553.

\bibitem{PhysRevA.89.052115}
L.~Rudnicki, Z.~Pucha\l{}a and K.~\ifmmode~\dot{Z}\else \.{Z}\fi{}yczkowski,
  {\em Phys. Rev. A} {\bf 89} (2014) 052115.

\bibitem{PhysRevA.54.3824}
C.~H. Bennett, D.~P. DiVincenzo, J.~A. Smolin and W.~K. Wootters, {\em Phys.
  Rev. A} {\bf 54} (1996) 3824.

\bibitem{PhysRevA.78.032332}
M.~Li, S.~M. Fei and Z.~X. Wang, {\em Phys. Rev. A} {\bf 78} (2008) 032332.

\bibitem{GRONDALSKI2002573}
J.~Grondalski, D.~Etlinger and D.~James, {\em Phys. Lett. A} {\bf 300}
  (2002) 573.

\end{thebibliography}
\end{document}